\documentclass[a4paper, 12pt]{article}
\usepackage[sort&compress]{natbib}
\bibpunct{(}{)}{;}{a}{}{,} 

\usepackage{amsthm, amsmath, amssymb, bbm,bbold,mathrsfs, multirow, url, subfigure}
\usepackage{graphicx} 
\usepackage{ifthen} 
\usepackage{amsfonts}
\usepackage[usenames]{color}
\usepackage{fullpage}

\theoremstyle{plain} 
\newtheorem{thm}{Theorem}

\newtheorem{lem}{Lemma}

\theoremstyle{definition}

\theoremstyle{remark}
\newtheorem{example}{Example}

\newcommand{\RR}{\mathbb{R}}

\newcommand{\A}{\mathscr{A}}

\newcommand{\E}{\mathsf{E}}

\newcommand{\bigmid}{\; \Bigl\vert \;}

\newcommand{\ptrue}{p^\star}
\newcommand{\Ptrue}{P^\star}

\newcommand{\N}{\mathsf{N}}
\newcommand{\nm}{\mathsf{N}}
\newcommand{\bet}{\mathsf{Beta}}

\newcommand{\gam}{\mathsf{Gam}}
\newcommand{\dpp}{\mathsf{DP}}

\newcommand{\dir}{\mathsf{Dir}}

\usepackage{graphicx}
\usepackage{amssymb}
\usepackage{epstopdf}
\usepackage{wrapfig}
\usepackage{lscape}
\usepackage[figuresright]{rotating}
\usepackage[normalem]{ulem}

\def\d{\mbox{d}}

\begin{document}

\title{On recursive Bayesian predictive distributions}
\author{P. Richard Hahn\footnote{Booth School of Business, The University of Chicago, {\tt richard.hahn@chicagobooth.edu}}, \; Ryan Martin\footnote{Department of Statistics, North Carolina State University, {\tt rgmarti3@ncsu.edu}}, \; and \; Stephen G. Walker\footnote{Department of Mathematics, University of Texas at Austin, {\tt s.g.walker@math.utexas.edu}}}
\date{\today}

\maketitle

\begin{abstract}
A Bayesian framework is attractive in the context of prediction, but a fast recursive update of the predictive distribution has apparently been out of reach, in part because Monte Carlo methods are generally used to compute the predictive.  This paper shows that online Bayesian prediction is possible by  characterizing the Bayesian predictive update in terms of a bivariate copula, making it unnecessary to pass through the posterior to update the predictive.  In standard models, the Bayesian predictive update corresponds to familiar choices of copula but, in nonparametric problems, the appropriate copula may not have a closed-form expression.  In such cases, our new perspective suggests a fast recursive approximation to the predictive density, in the spirit of Newton's predictive recursion algorithm, but without requiring evaluation of normalizing constants.  Consistency of the new algorithm is shown, and numerical examples demonstrate its quality performance in finite-samples compared to fully Bayesian and kernel methods. 

\smallskip

\emph{Keywords and phrases:} Copula; density estimation; nonparametric Bayes; prediction; recursive estimation.
\end{abstract}

\section{Introduction}  

Predictive distributions play a prominent role in Bayesian theory; in fact, sequences of predictive densities fully characterize a Bayesian model via the well known de~Finetti representation theorem, as discussed in de~Finetti (1937) and Hewitt and Savage (1955).  The Bayesian predictive density is obtained by updating the prior to the posterior and then marginalizing over the model parameters.  In particular, if $f(y \mid \theta)$ is the statistical model for iid real-valued data $Y_1,\ldots,Y_n$ and $\pi$ is the prior distribution for the parameter $\theta$, then the predictive density for $Y_{n+1}$, given $(y_1,\ldots,y_n)$, is given by
\begin{equation}
p_n(y)=\int f(y \mid \theta)\,\pi_n(\d \theta) \label{pred}
\end{equation}
where the posterior distribution $\pi_n$ for $\theta$, given $(y_1,\ldots,y_n)$, is
$$\pi_n(\d\theta)=\frac{\prod_{i=1}^n f(y_i \mid \theta) \,\pi(\d\theta)}{\int \prod_{i=1}^n f(y_i \mid \theta')\,\pi(\d\theta') }.$$
Intuitively, this Bayesian approach should be ideally suited to the accurate and coherent updating of information. However, by examining (\ref{pred}), one can see that there is no obvious route to quickly update the predictive directly: when a new observation is received, one first updates the posterior and then computes the integral to obtain the predictive.  This can be especially prohibitive when Monte Carlo methods are needed to compute the posterior.  The goal of this paper is to show that the Bayesian predictive distribution update can, indeed, be expressed in a recursive form, making fast online Bayesian prediction possible, even in complex nonparametric models.  

To show that the Bayesian predictive $p_n$ in \eqref{pred} can be updated without directly passing through the posterior, alleviating the need for Monte Carlo methods in Bayesian prediction, our starting point is a new observation that the predictive updates can be expressed in terms of a sequence of bivariate copula densities (e.g., Nelson 1999).  This observation is interesting for at least three reasons:
\begin{itemize}
\item according to de~Finetti's representation theorem, this sequence of copula densities provides an alternative characterization of the Bayesian model itself; 
\vspace{-2mm}
\item in cases where this sequence of copula densities can be identified analytically, this representation provides fast recursive updates to the Bayesian predictive; 
\vspace{-2mm}
\item and, even in cases where the sequence of copula densities cannot be written down analytically, the copula representation provides new insights on how to approximate the recursive updates.  
\end{itemize} 

The latter point above leads to the main contribution of the paper.  Many applications require both flexible nonparametric modeling and fast online estimation (Caudle and Wegman 2009). Such applications include  color modeling and tracking of objects (e.g., Elgammal et al.~2003; Han et al.~2008), finance (Lambert et al.~1999), network security, remote sensing. A particularly notable application area is the use of the Twitter data stream to make real-time predictions (Gerber 2014). However, the challenges in updating the Bayesian predictive are most acute in nonparametric problems, so kernel-based densities estimates (e.g., Raykar et al.~2010; Nakamura and Hasegawa 2013) are often preferred over Bayesian methods in these applications. While Bayesian methods have been used in parametric online prediction problems (e.g. the dynamic state-space models of West et al.~1985), their adoption in analogous nonparametric settings has been limited by extreme computational demands (Creal~2012). Bayesian computational methods, even those geared towards sequential analysis (e.g. Drovandi et al.~2013, Polson et al.~2010), do not focus on the {\em predictive distribution} directly, and therefore devote considerable resources to the computation of a posterior distribution over parameters. The approach in this paper will be to compute the predictive distribution directly. 

Perhaps the most commonly used Bayesian nonparametric model is the mixture of Dirichlet processes (e.g., Escobar 1988; Escobar and West 1995), but the need for Markov chain Monte Carlo methods to compute the posterior motivated Newton and Zhang (1999) and Newton (2002) to propose a \emph{predictive recursion} algorithm for estimating the posterior; see, also, Martin and Ghosh (2008), Tokdar et al.~(2009), and Martin and Tokdar (2009, 2011).  Despite its name, the predictive recursion algorithm is not fully satisfactory for estimating the predictive distribution: it targets the posterior instead of the predictive, so integration is needed to compute normalizing constants, etc.  Our copula characterization of the predictive update remains valid in nonparametric problems, but it may not be possible to derive the sequence of copula densities in closed-form.  It does, however, suggest a new version of the predictive recursion algorithm that targets the predictive density directly, avoiding the difficult problem of computing normalizing constants.  Besides being intuitively clear and fast to compute, we show both theoretically and numerically the accuracy of our proposed recursive predictive density estimate.   

The layout of the paper is as follows.  In Section~2 we provide the details of our representation of the predictive via copula models and identify the particular sequence of copula densities for some common Bayesian models.  Our investigation of the mixture of Dirichlet processes model lays the foundation for our recursive algorithm that directly targets the predictive densities presented in Section~3.  Numerical examples given in Section~4 demonstrate that the recursive copula approach is competitive with other common density estimation methods, i.e., Bayesian Gaussian mixture models and kernel density estimation methods, on prediction tasks.  In Section~5, we establish Kullback--Leibler consistency of the predictive distribution sequence.  Section~6 provides some concluding remarks and Appendices~A--B provide some technical proofs and other details about the recursive algorithm.

\section{A new look at Bayesian predictive updates} 

\subsection{Characterizing the updates via copula densities}

To characterize the Bayesian predictive updates, we take a sequential point of view.  That is, if $p_{n-1}$ is the predictive density for $Y_n$ based on observations $(y_1,\ldots,y_{n-1})$, then we want an update $(p_{n-1}, y_n) \mapsto p_n$ for the predictive density for $Y_{n+1}$ based on observations $(y_1,\ldots,y_n)$.  Consider the bivariate function $k(y,y')$ that satisfies 
\begin{equation}
p_n(y)=p_{n-1}(y)\,k(y,y_n).\label{update}
\end{equation}
Therefore, 
$$k(y,y_n)=\frac{p_n(y)}{p_{n-1}(y)}$$
which is symmetric in $(y,y_n)$, since
\begin{equation}
k(y,y_n)=\frac{\int f(y \mid \theta) \, f(y_n \mid \theta)\,\pi_{n-1}(\d\theta)}{\int f(y \mid \theta)\,\pi_{n-1}(\d\theta) \, \int f(y_n \mid \theta)\,\pi_{n-1}(\d\theta)}.\label{kstuff}
\end{equation}
The function $k(y,y_n)$ in (\ref{kstuff}) is easily seen to be a bivariate copula density function; that is, for some symmetric copula density $c_n$, which depends only on the sample through the sample size, we have
\begin{equation}
k(y,y_n)=c_n\left(P_{n-1}(y),P_{n-1}(y_n)\right)\label{copula0}
\end{equation}
where $c_n(u,v)=c_{n}(v,u)$ is a symmetric copula density, and $P_{n-1}$ is the distribution function corresponding to the predictive density $p_{n-1}$. 

We can now write the update $(p_{n-1}, y_n) \mapsto p_n$ as 
\begin{equation}
p_n(y) = c_n(P_{n-1}(y),P_{n-1}(y_n)) \, p_{n-1}(y) \label{copula}
\end{equation}
and for each Bayesian model there is a unique sequence $c_n$. Now (\ref{copula}) allows for the direct update of the predictive and moreover it can be seen that all one needs to direct a sequence of predictive densities is to define a sequence of copula functions $c_n$, the key to which is that $c_n \to 1$ as $n \to \infty$, i.e. the sequence of copula converges to the independent copula as the sample size increases.  

To put this all into context, the de~Finetti characterization of a Bayesian model is in terms of a (dependent) joint distribution over all future observables $p(y_1, y_2, y_3, \dots)$ and such a joint distribution can always be expressed in compositional form $p(y_1)p(y_2 \mid y_1)p(y_3 \mid y_1, y_2)\dots$. Additionally, Sklar's theorem (Sklar 1959) tells us that any joint distribution can be represented in copula form. These elements are familiar.
This paper focuses on the computational properties of a copula representation for the bivariate conditional distribution $p(y_n, y_{n+1} \mid y_{n-1}, \dots y_{1})$, as given in (\ref{copula0}), which will lead to a novel approximation of the predictive update in (\ref{copula}).

%
%

\subsection{Parametric model examples}

In this section we consider some standard Bayesian models, focusing on identifying the corresponding sequence $c_n$ of copula densities that characterizes the predictive updates. 

\begin{example}[Exponential model]
Here we consider the model and prior as $f(y \mid \theta) = \theta e^{-\theta y}$ and $\pi(\theta) = e^{-\theta}$, respectively.  Then standard calculations give 
$$p_{n-1}(y)=n\,\frac{T_{n-1}^n}{(T_{n-1}+y)^{n+1}},$$
where $T_{n-1}=1+y_1+\cdots+y_{n-1}$, 
and
$$p_n(y)=(n+1)\,\frac{(T_{n-1} + y_n)^{n+1}}{(T_{n-1} + y_n + y)^{n+2}}.$$
Therefore,
$$k(y,y_n)=\frac{(n+1)\,(T_{n-1}+y_n)^{n+1}\,(T_{n-1} + y)^{n+1}}{n\,T_{n-1}^n\,(T_{n-1} + y_n + y)^{n+2}},$$
which can be seen to be symmetric in $(y,y_n)$. Now
$$1-P_{n-1}(y)=\Bigl(\frac{T_{n-1}}{T_{n-1}+y}\Bigr)^n$$
and so $y=T_{n-1}\,\bigl[(1-P_{n-1}(y))^{-1/n}-1\bigr]$.  Therefore,
$$k(y,y_n)=\frac{n+1}{n}\,\frac{\{1-P_{n-1}(y)\}^{-(n+1)/n}\,\{1-P_{n-1}(y_n)\}^{-(n+1)/n}  }{\bigl[\{1-P_{n-1}(y)\}^{-1/n}+\{1-P_{n-1}(y_n)\}^{-1/n}-1\bigr]^{n+2} }$$
and so we have the Clayton copula (Clayton 1978), i.e.,
\[ c_n(u,v)=\frac{n+1}{n}\,\frac{(1-u)^{-1-1/n}\, (1-v)^{-1-1/n}}{\{(1-u)^{-1/n}+(1-v)^{-1/n}-1\}^{n+2}}, \]
with parameter $n^{-1}$, describing the sequence of predictive distributions.  Note that, as $n \to \infty$, $c_n$ converges to the independence copula.  
\end{example}

The calculations in Example~1 can be generalized to cover an exponential family model with conjugate prior, i.e., $f(y \mid \theta)=\xi(y)\,e^{y\,\theta-b(\theta)}$ and $\pi(\theta)\propto e^{\lambda\theta-\tau b(\theta)}$.  A by-product of this argument is the identification of a new and general class of copula that contains the Archimedean class.  Details are provided in Appendix~B.1.  

\begin{example}[Normal model]
Here we consider a normal model $f(y \mid \theta)=\N(y \mid \theta,1)$ and a conjugate prior $\pi(\theta)=\N(\theta \mid 0,\tau^{-1})$.  We claim that the predictive updates are characterized by a Gaussian copula with correlation parameter $\rho_n=(n+\tau)^{-1}$. In particular, we claim that the $c_n$ in \eqref{copula} is the Gaussian copula density $c_{\rho_n}$, where 
\begin{equation}
\label{eq:gauss.copula}
c_\rho(u,v) = \frac{\nm_2(\Phi^{-1}(u), \Phi^{-1}(v) \mid 0, 1, \rho)}{\nm(\Phi^{-1}(u) \mid 0, 1) \nm(\Phi^{-1}(v) \mid 0, 1)}, 
\end{equation}
with $\nm_2(\cdot \mid 0, 1, \rho)$ the standard bivariate normal density, with correlation $\rho$, and $\Phi$ the $\nm(0,1)$ distribution function.  To see this, start with the known form for the predictive, 
$$p_{n-1}(y)=\N\Bigl(y \bigmid \frac{T_{n-1}}{n-1+\tau},\frac{n+\tau}{n-1+\tau}\Bigr),$$
where $T_{n-1}=y_1 + \cdots + y_{n-1}$.  If we set $\mu_n=T_{n-1}/n$ and $\sigma_n^2=(n+\tau)/(n-1+\tau)$, then we have
$$P_{n-1}(y)=\Phi\Bigl(\frac{y-\mu_n}{\sigma_n}\Bigr). $$
Then the ratio $p_n(y) / p_{n-1}(y)$ is exponential and the key term in the exponent is 
$$\Bigl(y-\frac{y_n+T_{n-1}}{n+\tau}\Bigr)^2\,\frac{n+\tau}{n+1+\tau}-\Bigl(y-\frac{T_{n-1}}{n-1+\tau}\Bigr)^2\,\frac{n-1+\tau}{n+\tau}.$$
Next, using the fact that $\Phi^{-1}(P_{n-1}(y))=(y-\mu_n)/\sigma_n$, the key term in the exponent of $c_{\rho_n}(P_{n-1}(y), P_{n-1}(y_n))$ is  
$$\frac{\rho_n^2}{1-\rho_n^2}\Bigl[\Bigl(\frac{y-\mu_n}{\sigma_n}\Bigr)^2+\Bigl(\frac{y_n-\mu_n}{\sigma_n}\Bigr)^2\Bigr]-\frac{2\rho_n}{1-\rho_n^2}\Bigl(\frac{y-\mu_n}{\sigma_n}\Bigr)\,\Bigl(\frac{y_n-\mu_n}{\sigma_n}\Bigr).$$
The expressions in the two previous displays are equal up to constant terms when $\rho_n=(n+\tau)^{-1}$, which proves the claim.  Note that if the model were $f(y \mid \theta) = \nm(y \mid \mu, \sigma^2)$, with $\theta=(\mu,\sigma^2)$, and we put a standard conjugate prior on the variance parameter $\sigma^2$, then we would recover the Student-t copula for the update.
\end{example}

\begin{example}[Multinomial model]
Consider a multinomial model where there are $M$ categories and $f(y \mid \theta) = \theta_y$, where $\theta = (\theta_1,\ldots,\theta_M)$ is a probability vector.  Take a conjugate prior $\theta \sim \dir(\alpha_1,\ldots,\alpha_M)$, where each $\alpha_y$ is non-negative.  For data $y_1,\ldots,y_n$, let $T^n$ be the frequency table, with $T_y^n$ denoting the number of observations equal to $y$, $y \in \{1,\ldots,M\}$.  Using the standard theory for the multinomial--Dirichlet model, the predictive distribution $p_n$ is given by 
\[ p_n(y) = \frac{T_y^n + \alpha_y}{n + \beta}, \quad y \in \{1,\ldots,M\}, \]
where $\beta = \sum_{j=1}^M \alpha_j$.  From here, we can easily recover the predictive density ratio in \eqref{kstuff}:
\[ k(y, y_n) = \frac{p_n(y)}{p_{n-1}(y)} = \frac{n-1 + \beta}{n + \beta} \Bigl\{ 1 + \frac{1(y=y_n)}{T_y^{n-1} + \alpha_y} \Bigr\}. \]
To see what copula the predictive update corresponds to, we need to convert to the distribution function scale to find the function $C_n$ such that
\begin{align*}
C_n(P_{n-1}(y), P_{n-1}(y_n)) & = \sum_{z \leq y, \, z' \leq y_n} k(z, z') p_{n-1}(z) p_{n-1}(z') \\
& = (1-w_n) \, P_{n-1}(y) P_{n-1}(y_n) + w_n \, P_{n-1}(y) \wedge P_{n-1}(y_n), 
\end{align*}

\noindent where $w_n = (n + \beta)^{-1}$ and $x \wedge y = \min\{x,y\}$.  Thus, $C_n$ is a mixture of the Frechet--Hoeffding copula, $C_M(u,v)=u \wedge v$, and the independence copula, $C_I(u,v)=u\,v$.  Note that the $1-w_n$ weight assigned to independence copula converges to 1 as $n \to \infty$.  
\end{example}

\ifthenelse{1=1}{}{
\begin{example}[Multinomial model]
{\color{red} Here we consider the model where $y \in\{1,\ldots,M\}$ and $f(y \mid \theta)=\theta_y$ with $\theta \sim \dir(\alpha_1,\ldots,\alpha_M)$. This model is well known and 
$$c(j,j')=\frac{\E (\theta_j\,\theta_{j'})}{\E (\theta_j)\,\E (\theta_{j'})} = \begin{cases} q & \text{if $j \neq j'$} \\ q(1 + q/\alpha_j) & \text{if $j=j'$}, \end{cases} $$
where $q=b/(1+b)$ and $b=\sum_{j=1}^M \alpha_j$. Therefore, to recognize this copula, we find the copula distribution 
$$C(i,j)=\sum_{\ell \leq i,k \leq j} c(\ell,k)\,\E (\theta_\ell)\,\E (\theta_k)$$
which in more detail is
\begin{equation*}
\begin{split}
C(i,j) &=\frac{1}{b(b+1)}\Bigl\{\sum_{\ell=1}^{i \wedge j}\alpha_\ell+\sum_{\ell \leq i,k\leq j}\alpha_\ell\,\alpha_k \Bigr\} =\frac{1}{b+1}\,A_i \wedge A_j +\frac{b}{b+1}\,A_i\,A_j,
\end{split}
\end{equation*}
where $A_i=b^{-1}\sum_{\ell \leq i}\alpha_\ell$ and $x \wedge y = \min(x,y)$.  Thus $C(i,j)$ is a mixture of the Frechet--Hoeffding copula, $C_M(u,v)=u \wedge v$, and the independence copula, $C_I(u,v)=u\,v$.}
\end{example}
}

\subsection{A nonparametric model example}
\label{SS:dpmix}

Here we consider a nonparametric model, namely, a mixture of Dirichlet processes model as considered in Escobar (1988) and Escobar and West (1995), given by
$$f(y,G)=\int K(y \mid \theta)\,\d G(\theta),$$
where $K(y \mid \theta)$ is a given kernel and the prior assigned to $G$ is a Dirichlet process prior $\dpp(c,G_0)$, where $G_0$ is the base measure and $c > 0$ is the precision parameter (Ferguson, 1973).  This model was first introduced in Lo (1986) and the constructive definition of the Dirichlet process, see Sethuraman (1994), means we can write
$$f(y,G)=\sum_{j=1}^\infty w_j\,K(y \mid \theta_j),$$
where the $(\theta_j)$ are iid $G_0$ and the weights $(w_j)$ follow a stick-breaking construction, i.e., $w_1 = v_1$ and, for $j > 1$, $w_j=v_j\prod_{\ell<j}(1-v_\ell)$, with $(v_j)$ iid $\bet(1,c)$.  Hjort et al.~(2010) give details on this model and inference procedures using Markov chain Monte Carlo.

Let us assume that $K(y \mid \theta)=\N(y \mid \theta,1)$ and $G_0$ is $\N(0,\tau^{-1})$, as in Example~2. We can extend this to include a prior on the variance and we will recover the Student-t copula instead of the Gaussian copula.  Now, for the first update, we can compute the copula density; it is given by
\begin{equation}
\frac{\E \{f(y,G)\,f(y_1,G)\}}{p_0(y)\,p_0(y_1)},
\end{equation}
where $p_0(y)=\int K(y \mid \theta)\,\d G_0(\theta)$ is a $\mbox{N}(0,1+\tau^{-1})$ density, and
\begin{equation}
\E\{f(y,G)\,f(y_1,G)\}=\alpha \int K(y \mid \theta)\,K(y_1 \mid \theta)\,\d G_0(\theta) + (1-\alpha)\,p_0(y) \, p_0(y_1),
\end{equation}
and $\alpha = \sum_{j=1}^\infty \E(w_j^2)$.  Hence, the copula is a mixture of the Gaussian copula, $c_{\rho_0}$, in \eqref{eq:gauss.copula} with $\rho_0$ as in Example~2, and the independence copula.  Rewriting to explicitly highlight the copula representation yields
\begin{equation}\label{approx}
p_1(y)=(1-\alpha)\,p_0(y)+\alpha \,p_0(y)\,c_{\rho_0}\bigl( P_0(y),P_0(y_1) \bigr).
\end{equation}
Note that when written in this form, $p_0(y)$ need not be Gaussian any longer to define a valid update; the assumption of the Gaussian kernel $K(y \mid \theta)$ is reflected in the form of $c_{\rho_0}$, and $p_0(y)$ can be any choice of density function. One can think of this as first transforming ones data to standard normal and then applying the Bayesian update corresponding to the Dirchlet process model.

While it is not straightforward to extend the above derivation to a general update from $p_{n-1} \to p_n$ our strategy will be to iteratively apply (\ref{approx}) at each step, analogous to the approach of Newton for recursively approximating the posterior distribution; here we apply this idea directly to predictive distributions.

\section{Nonparametric recursive predictive distribution} 


Motivated by the calculations for the mixture of Dirichlet processes model in Section~\ref{SS:dpmix}, we propose the following recursive algorithm for directly updating the predictive, completely avoiding the posterior.  In particular, fix an initial guess $P_0$, with density $p_0$, and a sequence of weights $(\alpha_n) \subset (0,1)$.  Then, sequentially compute 
\begin{equation}
p_n(y)=(1-\alpha_n)\,p_{n-1}(y)+\alpha_n\, p_{n-1}(y)\,c_\rho(P_{n-1}(y),P_{n-1}(y_n)), \quad n \geq 1. \label{mdp}
\end{equation}
where $c_\rho$ is the Gaussian copula density in \eqref{eq:gauss.copula}.  The sequence $(\alpha_n)$ is based on stick breaks which are iid $\bet(1,c+n-1)$.  Therefore, they look like roughly $n^{-1}$, which is effectively what Newton took them to be; see \eqref{eq:weight.sequence} below.  Note that $\sqrt{1-\rho^2}$ is analogous to a kernel density bandwidth setting; after pre-scaling the data, we find that $\rho = 0.95$ works well in practice. Also note that the copula formulation amounts to applying a Gaussian transformation at each step, before carrying out the $n=1$ Dirichlet update.


Here we make three remarks. First, in the Gaussian copula model in Example~2, the sample size was captured by $\rho_n$ but, in \eqref{mdp}, the $\rho$ is held fixed and the sample size is carried by $\alpha_n$. Indeed, it is $\alpha_n$ going to 0 that takes us to the independence copula.  Second, the coherence property enjoyed by the ``correct'' Bayesian update, i.e., 
\[ \int p_n(y) p_{n-1}(y_n) \,\d y_n = p_{n-1}(y), \]
comes at a price---it cannot be computed recursively.  On the other hand, by sacrificing this coherence, we can get a fast update which is still theoretically and numerically accurate.  To be clear, the update $p_0$ to $p_1$ is the exact Bayesian update and, therefore, must be good; our proposal is to replicate this ``good'' update for all $n$.  We lose the coherence property above, but gain computational efficiency; simulations (reported later) suggest that the resulting approximation of the predictive is satisfactory at various values of $n > 1$.  Third, although Newton's original algorithm can be used to compute an approximation to the predictive, there are difficulties due to the need to evaluate intractable normalizing constants.  Indeed, the $i^{\text{th}}$ step of Newton's original algorithm, which takes the previous estimate $G_{i-1}$ of the mixing distribution and the current observation $y_i$ to a new estimate $G_i$, requires evaluation of a normalizing constant $\int K(y_i \mid \theta) \, \d G_{i-1}(\theta)$, which cannot be computed analytically since $G_{i-1}$ is not of any standard form.  By working directly with the predictive, as we do here, there is no need to evaluate such normalizing constants.  

A few words should also be said about the implementation.  It is actually simpler to work on the distribution function scale, where the algorithm looks like 
\begin{equation}
P_n(y)=(1-\alpha_n)\,P_{n-1}(y)+\alpha_n\,H_\rho(P_{n-1}(y),P_{n-1}(y_n)). \label{cdf}
\end{equation}
where
\begin{equation}
\label{eq:H.rho}
H_\rho(u,v)=\Phi\Bigl(\frac{\Phi^{-1}(u)-\rho\,\Phi^{-1}(v)}{\sqrt{1-\rho^2}}\Bigr). 
\end{equation}
In this formulation it is evident that $P_n$ in \eqref{cdf} is a weighted average of $P_{n-1}$ and a suitable transformation of a normal distribution with variance $1-\rho^2$ and centered at $\rho \Phi^{-1}(P_{n-1}(y_n))$. As $\rho$ nears 1, this second term becomes a step distribution with single jump at $\Phi^{-1}(P_{n-1}(y_n))$. Intuitively, the method is similar to kernel density estimation, with two differences: iteratively applied adaptive transformations based on the current distribution estimate $\Phi^{-1}(P_{n-1}(\cdot))$, and shrinkage towards the prior predictive $P_{n-1}$.

Computationally, we take a fixed grid of points, $\{\bar y_m: m=1,\ldots,M\}$, in $\RR$ and compute the sequence $P_n(\bar y_m)$ for each $m$.  Then the distribution function $P_n(y)$ can be plotted by interpolation.  From this, the density $p_n(y)$ can be obtained by approximating the derivative by a difference ratio.  Given the distribution function or density evaluated on a fine grid of points, features of the predictive distribution, such as the mean or quantiles, can be readily obtained.

We conclude this section by giving an illustration of the recursive predictive distribution estimator for univariate data; a bivariate data example is presented in Appendix~A.  We compared to a Dirichlet process mixture of normals as well as a mixture of P\'olya trees. The example is taken from the {\tt R} package {\tt DPpackage} (Jara et al 2011).  Consider the well-known ``galaxy'' data of Roeder (1990) consisting of $n=82$ velocity measurements (in km/second) of galaxies obtained from an astronomical survey of the Corona Borealis region. Figure~1 shows three density estimates: a mixture of P\'olya trees, a Dirichlet process mixture of normals, a kernel density estimate, and the new recursive approximation. For this fit we use an empirical Bayes selection of the hyperparameters with $p_0$ a normal density with variance 9 and mean set to the mean of the data; we also take $\rho=0.95$ and $\alpha_i = (i+1)^{-1}$. The priors of the P\'olya tree and Dirichlet process model are set according to the demonstration code from the {\tt DPpackage}. Note that the Dirichlet process fit is quite close to the recursive approximation (the dashed versus the solid densities).

\begin{figure}
\begin{center}
\includegraphics[width=4.5in]{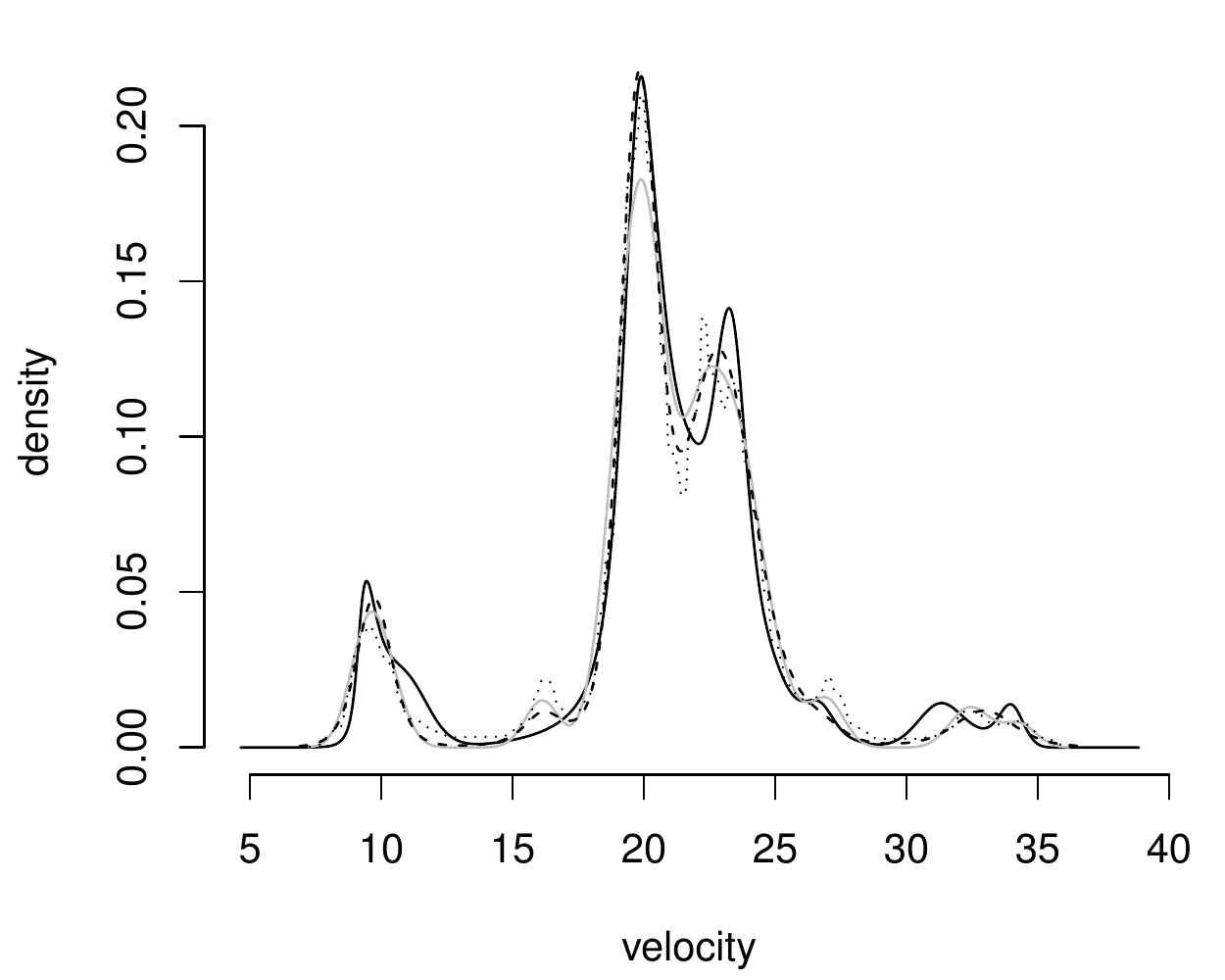}
\caption{With $n=82$ observations from Roeder (1990), the fits from a P\'olya tree mixture model (dotted), a Dirichlet process mixture of normals (dashed), the new recursive method (solid) give visually similar fits. A kernel density estimate is also shown in gray.}
\end{center}
\end{figure}

\section{Simulation studies}

\subsection{Setup}

Keeping with our focus on predictive distributions, this section evaluates the recursive estimator \eqref{cdf}, with $\rho=0.95$ and $\alpha_i = (i+1)^{-1}$, in terms of a predictive loss, measuring the difference between a prediction and a future realization of an observable variable.  Specifically, we consider a finite vector of quantiles, defining a vector valued check-loss function.  The check loss function is a piece-wise linear loss function which can be expressed as
\begin{equation}
\ell_q(y, a) = (1-q)(a-y)\mathbb{1}(y<a) + q\mathbb{1}(y > a)(y - a), \quad q\in(0,1).
\end{equation}
Check loss gets its name from the check-shaped graph of the function. 
Check loss can be justified intuitively in terms of asymmetric costs. To take a simple example, consider a restaurant: too much inventory leads to waste via spoilage at some cost per unit (purchase price), while too little inventory leads to foregone sales due to unfulfillable orders at a distinct cost per unit (because orders for multiple items are canceled in their entirety).  Check loss is intimately related to quantile estimation as follows: it is straightforward to show that for any density function $f(y)$ with distribution function $F(y)$, the integrated (expected) check loss is minimized at $F^{-1}(q)$.  In our simulation study, we use a vector-valued check loss function defined by a vector parameter $q$; specifically we consider $q = (0.001,\,0.01,\,0.1,\,0.25,\,0.5,\,0.75,\,0.9,\,0.99,\,0.999)$.

\subsection{Batch mode simulation study}

For our first simulation study, we compare the performance of our recursive approximation of the predictive density to that arising from the posterior of a Bayesian Dirichlet process mixture model, fit using the function {\tt DPdensity} as well as a P\'olya tree mixture model using the function {\tt PTdensity}, both from the {\tt R} package {\tt DPpackage} (Jara et al 2011).  Because the goal of our simulation is to compare the closeness of the approximation, all model hyperparameters were calibrated to replications of the data before the simulation study was started, to ensure that the model fits were not grossly inappropriate. Details of the model fitting are available in the authors' {\tt R} script.


We generate the data, $Y$, according to a two component mixture of t-distributions with 5 degrees of freedom. One of these components is fixed to have location parameter 1 and scale parameter 1.  The second component has mean $\mu$ and scale $s+1$. We simulate 500 independent samples from this distribution of size $n=50$. For each sample, the values of $\mu$, $s$ and the mixing proportion $w$, are drawn at random according to $w \sim \bet(2,2)$, $s \sim \gam(1,1)$ and $\mu \sim \nm(0,4)$. 

To evaluate each method, we compute the mean check loss on a Monte Carlo sample of size 100,000 from the true distribution, using the optimal action according to the inferred predictive distribution using each method, which we denote $a_{\text{recursive}}$ and $a_{\text{bayes}}$ respectively. We also compute $a_{\text{truth}}$ which is the check loss minimizer according to the true data generating distribution. Finally, we consider the scaled difference of integrated check loss: 
\begin{equation}
\Delta_q = \frac{\E\lbrace\ell_q(Y, a_{\text{recursive}})\rbrace - \E\lbrace\ell_q(Y, a_{\text{bayes}})\rbrace}{\E\lbrace\ell_q(Y, a_{\text{truth}})\rbrace}.
\end{equation}
We evaluate $\Delta_q^{(j)}$ for $j= 1\dots500$ trials and a range of $q$. The upshot is that for the ``easier'' quantiles, the three methods all agree nicely. There is greater discrepancy for very high and very low quantiles; it is notable, however, that the recursive update method gives better average loss on these quantiles, although the reason why is not clear. The comparisons to the Dirichlet process mixture of normals is given in Table~1; the comparison to the P\'olya tree mixture is given in Table~2.  This simulation study was conducted for various samples sizes, from 10 to 100, with qualitatively similar results.



\begin{table}
\begin{center}
\begin{tabular}{r|rrr}
$q$ &  Mean & Median & St. Dev. \\
 \hline
0.001 &  $-$10\% & 0\%& 56\%\\
0.01 & $-$26\% & 0\% & 50\% \\
0.10 & 0\% & 0\%& 2\% \\
0.25 & 0\% & 0\%& 1\% \\
0.50 & 0\% & 0\%& 1\%\\
0.75 & 0\% & 0\%& 1\%\\
0.90 & $-$0\% & 0\%& 3\%\\
0.99 & $-$3\% & 0\%& 10\%\\
0.999 & $-$8\% & 0\%& 90\%\\
 \end{tabular}
 \caption{Summary statistics of the distribution of $\Delta_q$ defined relative to the Dirichlet process mixture of Gaussians across 500 simulations for $n=50$ observations.}
\end{center}
\end{table}

\begin{table}
\begin{center}
\begin{tabular}{r|rrr}
$q$ &  Mean & Median & St. Dev. \\
 \hline
0.001 &  $-$3\% & 0\%& 67\%\\
0.01 & $-$22\% & 0\% & 63\% \\
0.10 & 0\% & 0\%& 3\% \\
0.25 & 0\% & 0\%& 1\% \\
0.50 & 0\% & 0\%& 1\%\\
0.75 & 0\% & 0\%& 2\%\\
0.90 & 0\% & 0\%& 3\%\\
0.99 & $-$3\% & 0\%& 27\%\\
0.999 & $-$70\% & $-$32\%& 159\%\\

 \end{tabular}
 \caption{Summary statistics of the distribution of $\Delta_q$ defined relative to the P\'olya tree mixture across 500 simulations for $n=50$ observations.}
 \label{results}

\end{center}
\end{table}

\subsection{Sequential simulation study}

Next, we consider online prediction according to the check loss function. That is, as individual observations arrive, we want to make an optimal action to be evaluated upon the subsequent observation. In this scenario, the extreme slowness of an MCMC approach precludes the use of the routine Gibbs sampled Gaussian mixture model, as this setting would demand rerunning the full sampling chain each time a new observation arrived. As such, our comparison method for this exercise is the Dirichlet process Gaussian mixture model particle learning method described in Carvalho, et al. (2010), which is, by construction, more computationally suited to the on-line setting. We do not provide the details of this method here.  Additionally, we compare to a kernel density estimator with bandwidth selected by the method of Sheather and Jones (1991). 

Note that the recursive bivariate copula approach is approximately as fast as the kernel density approach, with minor differences due to implementation specifics, such as what language the code is written in.  The particle filter approach, while much faster than MCMC, requires storing a great deal of additional information (the ``state vectors" of the filter) and, as a result, takes longer to compute. It should be mentioned that this additional overhead comes with a benefit, which is that the particle method gives full posteriors over model parameters in an online fashion; our approach bypasses those elements in order to directly compute the predictive and is faster as a result.

For this simulation we consider a sample of size $n = 50$ with data generated according to the same recipe as described in the previous section.  An initial four observations are used to ``prime" the predictive distributions;  observations are then introduced one-by-one and a check-loss-optimal prediction is made based on the posterior predictive at each time point, which is then evaluated at the subsequent observation, $j = 5,\ldots,50$. The aggregate check-loss over this period is then computed and stored. This process is repeated for 500 simulations. For this study we consider the tenth percentile, $q = 0.1$. 

For the recursive copula method we take 
$p_0$ a standard Cauchy distribution. We implement the particle learning algorithm using utility functions provided in the {\tt R} package {\tt Bmix} (Taddy 2010). We use default parameter values as given in provided one-dimensional density estimation demo in that package, with 200 particles. For the kernel density method, we re-estimate the bandwidth with every new observation.

As before, we consider the standardized difference:
\begin{equation}
\Delta_q = \frac{\sum_{j=5}^{50} \ell_q(y_j, a_{\text{recursive}}^{(j-1)}) - \sum_{j=5}^{50}\ell_q(y_j, a_{\text{particle}}^{(j-1)})}{\sum_{j=5}^{50}\ell_q(y_j, a_{\text{truth}})},
\end{equation}
where $a^{(k)}$ denotes the inferred optimal action after observing $k$ data points.

\begin{figure}
\begin{center}
\includegraphics[width=3in]{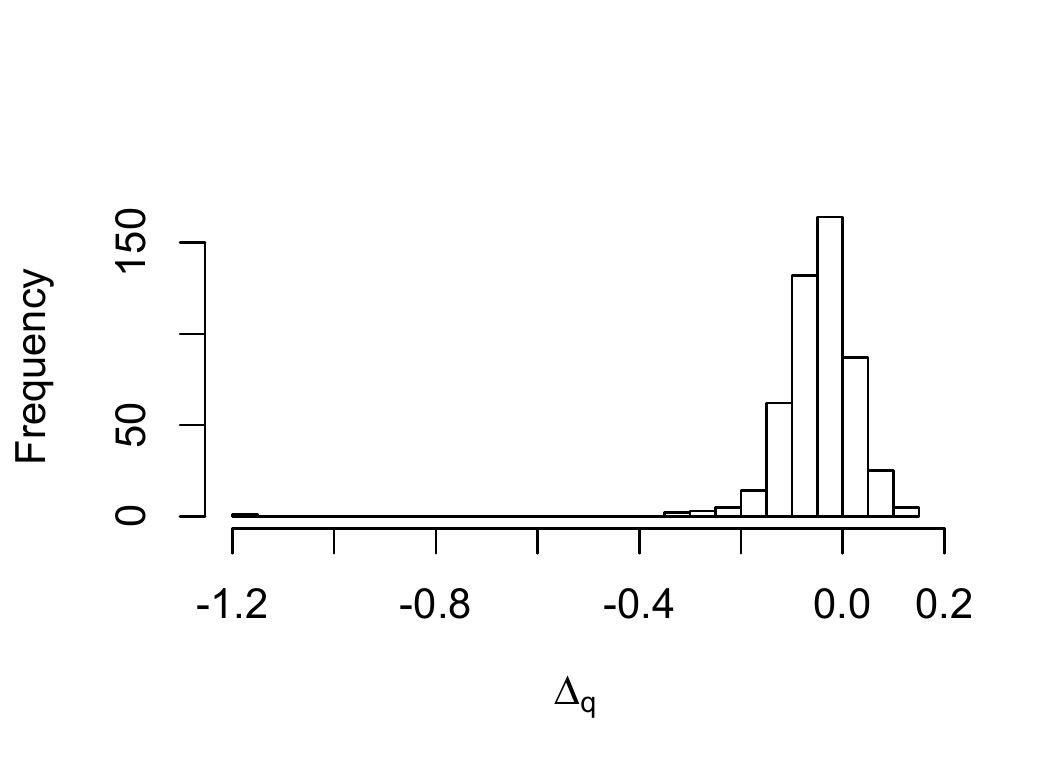}\includegraphics[width=3in]{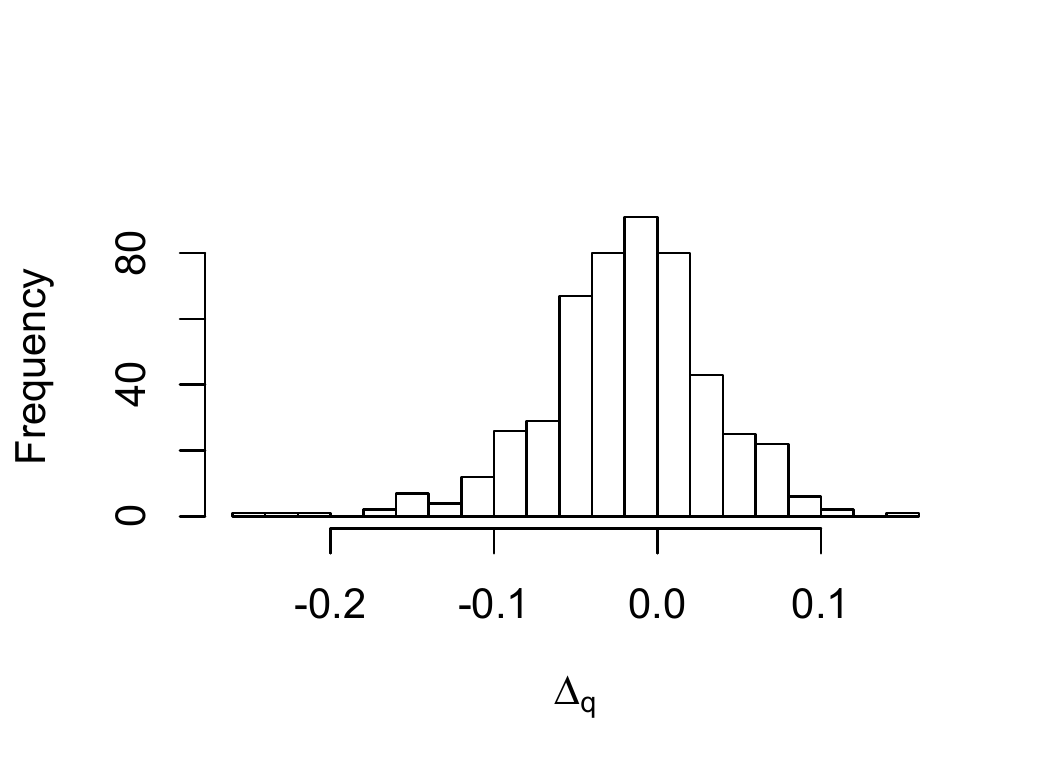}
\end{center}
\caption{Distribution of scaled difference in summed check loss ($q=0.10$) over 500 simulations, for $n = 50$ sequential observations. Negative numbers mean the recursive method outperformed the competing method. Units are in percentage of the theoretical optimal check loss. Left panel compares to Bayesian particle learning mixture model, right panel compares to kernel density estimate with adaptive bandwidth.}
\label{result}
\end{figure}

\begin{table}
\begin{center}
\begin{tabular}{r|rrr}
Comparison method &  Mean & Median & $\mbox{Pr}(\Delta_{0.1} < 0)$ \\
 \hline
Particle learning GMM &  $-$4.6\% & $-$4.3\%& 0.77\\
Kernel density estimate & $-$1.9\% & $-$1.5\% & 0.64 \\

 \end{tabular}
 \caption{Summary statistics of scaled difference in summed check loss ($q=0.10$) for two competing methods over 500 simulated data sets.}
\end{center}
\end{table}

As with the batch simulation study, our claim is not that the bivariate recursive method is outright superior to these alternatives. However, these simulations highlight certain virtues of the approach---speed and ease-of-implementation---while demonstrating that the performance is broadly competitive. Concretely: our method is ten times faster than the {\tt Bmix} package at computing the posterior predictive (with 200 particles) and our main function is 15 lines long, while the functions underlying {\tt Bmix} are many hundreds of lines long. 

It is worth emphasizing that simulation studies such as those reported here are inherently sensitive to prior specification:  after all, attempting to infer the 10th percentile based on only fifty observations is a difficult task that will benefit from wise choices of prior. That said, we argue that the recursive bivariate copula approach has an advantage in terms of being relatively transparent in terms of its prior specification (the initial distribution function $p_0$ can be a convenient parametric form) and its hyper-parameter $\rho$.
Mixture models of any kind do not boast this advantage.

\section{Asymptotic convergence theory}


The recursive algorithm is designed to approximate the posterior predictive under the Dirichlet process mixture model.  When the sample size is large, the posterior predictive agrees with the true data-generating distribution, so it makes sense to investigate the asymptotic convergence properties of the recursive estimator $P_n$ in \eqref{cdf} to the true distribution function $\Ptrue$.   Recall that the proposed algorithm is based on a Gaussian copula via the function $H_\rho$ in \eqref{eq:H.rho}, and throughout we take the copula correlation parameter $\rho \in (0,1)$ to be fixed.  We will also require that the weight sequence $(\alpha_n)$ satisfies
\begin{equation}
\label{eq:weight.sequence}
\alpha_n = a (n+1)^{-1}, \quad n \geq 1,
\end{equation}
for some sufficiently small $a > 0$; see Lemma~\ref{lem:bound} below.  This implies that 
\[ \sum_{i=1}^\infty \alpha_i = \infty \quad \text{and} \quad \sum_{i=1}^\infty \alpha_i^2 < \infty, \]
which is standard in the stochastic approximation literature (e.g., Kushner and Yin 2003).  

In this section, we prove that the recursive predictive distribution sequence $(P_n)$ converges to the true distribution $\Ptrue$ in the Kullback--Leibler sense, with probability~1 as $n \to \infty$.    Towards this, consider the algorithm for the predictive density $p_n(y)$, given by 
\begin{align*}
p_n(y) & = (1-\alpha_n) p_{n-1}(y) + \alpha_n p_{n-1}(y) \, c_\rho(P_{n-1}(y), P_{n-1}(Y_n)) \\
& = p_{n-1}(y)\bigl[ 1 + \alpha_n \{c_\rho(P_{n-1}(y), P_{n-1}(Y_n)) - 1\} \bigr], 
\end{align*}
where $c_\rho(u,v)$ is the bivariate Gaussian copula density \eqref{eq:gauss.copula} with correlation parameter $\rho > 0$ and $P_0$ is an initial guess.  Let $K$ denote the Kullback--Leibler divergence, and $\ptrue$ the true data-generating density; the goal is to show that $K(\ptrue, p_n) \to 0$ $\Ptrue$-almost surely.  Our analysis here is based on that in Martin and Tokdar (2009) for proving consistency of Newton's original predictive recursion algorithm.  However, since there is no natural mixture model structure, some new ideas are needed.  The main ingredient is a representation \eqref{eq:copula.mix} of the Gaussian copula density as a sort of mixture.  

To start, write
\begin{align*}
K(\ptrue, p_n) - K(\ptrue, p_{n-1}) & = -\int \log \frac{p_n(y)}{p_{n-1}(y)} \ptrue(y) \,\d y \\
& = -\int \log\bigl[ 1 + \alpha_n\{c_\rho(P_{n-1}(y), P_{n-1}(Y_n)) - 1\} \bigr] \ptrue(y) \,\d y.
\end{align*}
For $x$ away from $-1$, i.e., $x \approx 0$, the following inequality holds:
\[ \log(1+x) \geq x - 2x^2, \quad x \approx 0. \]
This inequality can be applied in our case, since $c_\rho \geq 0$ and $\alpha_n \to 0$, and it gives 
\[ K(\ptrue, p_n) - K(\ptrue, p_{n-1}) \leq -\alpha_n \int \{c_\rho(P_{n-1}(y), P_{n-1}(Y_n)) - 1\} \ptrue(y) \,\d y + R_n, \]
where the ``remainder term'' $R_n$ is given by
\[ R_n = 2\alpha_n^2 \int \{c_\rho(P_{n-1}(y), P_{n-1}(Y_n)) - 1\}^2 \ptrue(y) \,\d y. \]
Taking conditional expectation with respect to $\A_{n-1}=\sigma(Y_1,\ldots,Y_{n-1})$, we get 
\begin{align}
\E\{K&(\ptrue, p_n) \mid \A_{n-1}\} - K(\ptrue, p_{n-1}) \notag \\ 
& \leq -\alpha_n \int \int \{c_\rho(P_{n-1}(y), P_{n-1}(y')) - 1\} \ptrue(y) \ptrue(y') \,\d y \,\d y' + \E(R_n \mid \A_{n-1}). \label{eq:double}
\end{align}
If the double integral above is positive, and the remainder term is negligible, then $K_n := K(\ptrue, p_n)$ is an ``almost supermartingale'' (Robbins and Siegmund 1971) and converges to an almost sure limit, say, $K_\infty$.  To handle the double integral in \eqref{eq:double}, and to show that the limit is almost surely zero, some manipulation of the copula density $c_\rho$ is needed.  

Traditionally, the copula density is written as in Equation \eqref{eq:gauss.copula} above, which has a relatively simple closed-form expression that is used for practical implementation.  However, for our theoretical analysis, it will be convenient to rewrite the copula density as 
\begin{equation}
\label{eq:copula.mix}
c_\rho(u,v) = \int \psi_\theta(u) \psi_\theta(v) \nm(\theta \mid 0, \rho) \,\d\theta, 
\end{equation}
where $\psi_\theta$ is a ratio of normal densities, 
\[ \psi_\theta(u) = \frac{\nm(\Phi^{-1}(u) \mid \theta, 1-\rho)}{\nm(\Phi^{-1}(u) \mid 0, 1)}. \]
This follows from routine calculations using normal convolutions.  The point is that the Gaussian copula has a type of mixture or ``conditionally iid'' representation.

Kullback--Leibler consistency also requires two preliminary results; see Appendix~\ref{S:proofs} for the proofs.  For the first, write $T(p_n)$ for that double integral on the right-hand side of \eqref{eq:double}, i.e., 
\[ T(p) = \int \int \{c_\rho(P(y), P(y')) - 1\} \ptrue(y) \ptrue(y') \,\d y \,\d y', \]
where $p$ is a generic density with distribution function $P$.  If we plug in the alternative representation \eqref{eq:copula.mix} of the copula density into the formula for $T(p)$ and interchange the order of integration, we get 
\begin{align*}
T(p) & = \int \Bigl[ \Bigl\{ \int \psi_\theta(P(y)) \ptrue(y) \,\d y \Bigr\}^2 - 1\Bigr] \nm(\theta \mid 0, \rho) \,\d\theta \\
& = \int \Bigl\{ \int \psi_\theta(P(y)) \ptrue(y) \,\d y - 1\Bigr\}^2 \nm(\theta \mid 0, \rho) \,\d\theta, 
\end{align*}
where the last expression follows from the formula $\E(X^2) - \E^2(X) = \E\{X-\E(X)\}^2$ and the fact that $\int \psi_\theta(u) \nm(\theta \mid 0, \rho) \, \d\theta = 1$ for all $u$.  

\begin{lem}
\label{lem:positive}
Consider a density $p$ whose support contains that of $\ptrue$.  Then $T(p) \geq 0$ with equality if and only if $p=\ptrue$ Lebesgue-almost everywhere.  
\end{lem}

Our second preliminary result demonstrates that the remainder term $R_n$ is negligible, i.e., it vanishes sufficiently fast that it does not disrupt the supermartingale-like dynamics of Kullback--Leibler sequence $K_n = K(\ptrue, p_n)$.  

\begin{lem}
\label{lem:bound}
Write $\bar P_0 = 1-P_0$.  Suppose that $P_0$ and $\ptrue$ satisfy
\begin{equation}
\label{eq:integrable}
\int \{ P_0(y) \wedge \bar P_0(y) \}^{-2\rho/(1+\rho)} \ptrue(y) \,\d y < \infty. 
\end{equation}
Furthermore, assume that $a$ in \eqref{eq:weight.sequence} satisfies 
\[ 0 < a < \frac{2\rho + 2}{7\rho + 1} . \]
Then $\sum_n \E(R_n \mid \A_{n-1}) < \infty$ $\Ptrue$-almost surely.
\end{lem}

The integrability condition \eqref{eq:integrable} can be understood as a requirement that the recursive algorithm's initialization cannot be too light tailed compared to $\ptrue$; this is consistent with our choice in Section~4 to use a heavy-tailed $P_0$.  

\begin{thm}
\label{thm:limit}
Let $p_n$ be the predictive density for $Y_{n+1}$, given $Y_1,\ldots,Y_n$ defined above, with correlation parameter $\rho \in (0,1)$ and with weight sequence $(\alpha_n)$ that satisfies \eqref{eq:weight.sequence}.  If the true density $\ptrue$ is continuous and satisfies \eqref{eq:integrable} for the given $P_0$, 
then $K(\ptrue, p_n) \to 0$ $\Ptrue$-almost surely.  
\end{thm}

\begin{proof}
From the expression for $\E(K_n \mid \A_{n-1}) - K_{n-1}$, and Lemmas~\ref{lem:positive}--\ref{lem:bound}, it follows from Robbins and Siegmund (1971) that 
\[ K_n \to K_\infty \quad \text{and} \quad \sum_n \alpha_n T(p_n) < \infty, \quad \text{$\Ptrue$-almost surely}. \]
It remains to show that $K_\infty = 0$ $\Ptrue$-almost surely.  Suppose, to the contrary, that $K_\infty > 0$ with positive probability.  Then $p_n$ is away from $\ptrue$ (in the Kullback--Leibler sense) for all but finitely many $n$ with positive probability.  More precisely, there is a set of positive Lebesgue measure on which $p_n \neq \ptrue$.  By Lemma~\ref{lem:positive}, this implies $T(p_n) > 0$ for all but finitely many $n$.  Since $T(p_n)$ is bounded away from zero, we get $\sum_n \alpha_n T(p_n) = \infty$ with positive probability, which contradicts the second conclusion in the above display.  Therefore, we must have $K_\infty = 0$ almost surely, completing the proof.
\end{proof}

\section{Conclusion}

In this paper, we have identified an interesting new connection between Bayesian predictive updates and well-known bivariate copulas.  Besides the new light cast on this previously unknown connection between Bayesian inference and copulas, which can provide further and deeper insights and understanding about both, this development makes clear that Bayesian predictive updates do not require posterior computations.  This opens the door for online Bayesian prediction, as well as for Bayesian predictive analysis for researchers who are uncomfortable with the implementation and/or slow speeds of Markov chain Monte Carlo methods.  

The new recursive algorithm developed here is important because it provides a direct attack on the predictive density, which can simplify both the modeling and the computational aspects in applications.  First, if the predictive is the goal, then needing to specify a mixture model, especially, a support for the mixing distribution, is undesirable, and the new algorithm circumvents this.  Second, Newton's original algorithm requires computation of a normalizing constant at each iteration, and these are never available in closed-form.  For mixing distributions supported on one- or two-dimensional spaces this can easily be handled with quadrature but, to date, there is no efficient strategy for computing these normalizing constants for higher-dimensional spaces.  Again, the new version that directly attacks the predictive distribution avoids all of these difficulties.  


\section*{Acknowledgments}

The authors thank the Editor, Associate Editor, and referees for their helpful comments on the previous version of this manuscript.  This work is partially supported by the U.~S.~National Science Foundation, grants DMS--1507073 and DMS--1506879, and by the U.~S.~Army Research Offices, Award \#W911NF-15-1-0154.

\appendix
\section{Extension to bivariate data}
The recursive update applied to bivariate data is
\begin{equation}
\begin{split}
P_n(y,x) & = (1-\alpha_n)\,P_{n-1}(y,x) \\
& \qquad + \alpha_n\,H_\rho(P_{n-1}(y,x),P_{n-1}(y_n,x)) H_\rho(P_{n-1}(y,x),P_{n-1}(y,x_n)). 
\end{split}
\end{equation}
so long as the initial kernel $P_0$ is independent so that $P_0(x,y) = P_0(x)P_0(y)$. 

Here we demonstrate this update applied to air quality data, an example taken again from the {\tt R} package {\tt DPpackage} (Jara et al 2011). As before, we compare to both a Dirichlet process mixture of normals, as well as a P\'olya tree mixture.  The data are daily ozone and solar radiation measurements taken in New York between May through September in 1973 as recorded by the New York State Department of Conservation and the National Weather Service, respectively. The data consist of 153 pairs, but only 111 have both values, and these are used in the fits shown in Figure~3. The displayed recursive approximation is a point-wise average of ten random permutations of the data. Visually, note that the recursive approximation estimate is smoother than the P\'olya tree estimate, but that its overall shape is more similar to the P\'olya tree estimate than to the Dirichlet process estimate. Notably, the recursive approximation took seconds to fit, while the other methods required minutes-long MCMC runs.

\begin{figure}
\begin{center}
\includegraphics[width=6in]{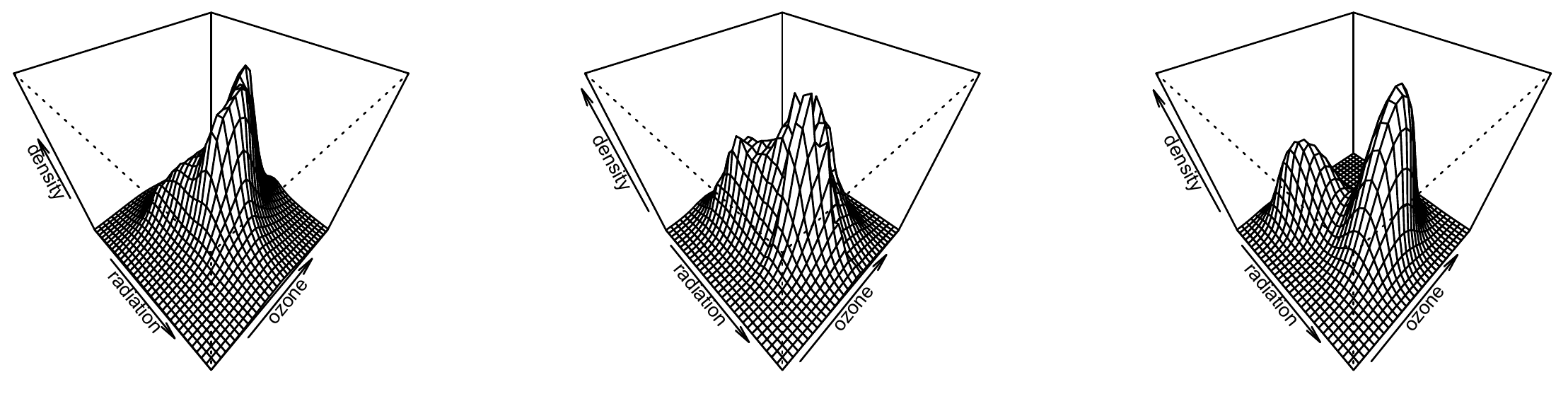}
\caption{Fit to 111 air quality measurements, the  recursive approximation (left), a P\'olya tree mixture model (center), and a Dirichlet process mixture of normals (right), give broadly similar fits, although not as similar as in the one-dimensional example.}
\end{center}
\end{figure}

\section{Technical details and proofs}

\subsection{Generalization of Example~1}
\label{S:expfam}

Write $F(\cdot)=F(\cdot \mid\theta)$ for the distribution function corresponding to the model density $f(\cdot \mid \theta)$.  Suppose that $F$ can be written as
$$F(y)=\int^y_{-\infty} \xi(s)\,\phi(s)\,\d s$$
for non-negative functions $\xi$ and $\phi$.  Then a symmetric copula distribution can be defined via
$$C(u,v)=\int^{G(u)}_{-\infty}\int^{G(v)}_{-\infty} \xi(s)\,\xi(t)\,\widetilde{\phi}(t+s)\,\d s\,\d t,$$
for some function $\widetilde{\phi}$, where $G=F^{-1}$, provided
\begin{equation}
\int^\infty_{-\infty} \xi(s)\,\widetilde{\phi}(t+s)\,\d s =\phi(t).\label{condi}
\end{equation}
Define
$$d_n(z)=\int \exp\{(z+\lambda)\theta-(n+\tau)b(\theta)\}\,\d\theta,$$
so
$$P_{n-1}(y)=\int^y_{-\infty} \xi(s)\,\frac{d_{n}(s+\mu)}{d_{n-1}(\mu)}\,\d s,$$
where $\mu=\sum_{i=1}^{n-1}y_i$.  Also,
$$C(P_{n-1}(x),\,P_{n-1}(y))=\int^x_{-\infty}\int^y_{-\infty} \xi(s)\,\xi(t)\frac{d_{n+1}(s+t+\mu)}{d_{n-1}(\mu)}\,\d s\,\d t.$$
The aim now is to show that for the functions $\phi = \phi_{n,\mu}$ and $\widetilde\phi=\widetilde\phi_{n,\mu}$ with 
$$\phi(s)=\frac{d_{n}(s+\mu)}{d_{n-1}(\mu)}\quad\mbox{and}\quad \widetilde{\phi}(s+t)=\frac{d_{n+1}(s+t+\mu)}{d_{n-1}(\mu)}$$
we get that $(\phi,\widetilde{\phi})$ satisfy (\ref{condi}).  To this end, we have that
$$\int^\infty_{-\infty} \xi(t)\,d_{n+1}(t+s+\mu)\,\d t$$ 
is given by
$$
\int^\infty_{-\infty} \int \xi(t)\,\exp\{t\theta-b(\theta)\}\,
\exp\{(s+\mu+\lambda)\theta-(n+\tau)b(\theta)\}\,\d\theta$$
which is easily seen to be $d_n(s+\mu)$, and hence (\ref{condi}) holds.

Thus, the new class of symmetric copula is given by
$$C(u,v)=\int_{-\infty}^{G(u)}\int_{-\infty}^{G(v)}\xi(s)\,\xi(t)\,\widetilde{\phi}(t+s)\,\d s\,\d t,$$
with 
$$\int_{-\infty}^\infty \xi(s)\,\widetilde{\phi}(s+t)\,\d s=G^{-1}(t).$$
The corresponding copula density function is given by
$$c(u,v)=g(u)\,g(v)\,\xi(G(u))\,\xi(G(v))\,\widetilde{\phi}(G(u)+G(v)),$$
where $g=G'$. To see this more clearly, let $\xi \equiv 1$, so
$$c(u,v)=\frac{\widetilde{\phi}(G(u)+G(v))}{\phi(G(u))\,\phi(G(v))}$$
and now $\phi$ is the density function corresponding to $F$. If $\phi'$ is negative, i.e. $F$ is concave, then $\widetilde{\phi}=(-F)''$, giving
$$c(u,v)=\frac{(-F)''(F^{-1}(u)+F^{-1}(v))}{F'(F^{-1}(u))\,F'(F^{-1}(v))}.$$
This is an Archimedean copula and hence the new class of copula provides a generalization. Note that for the exponential model considered in Example 1, we did have $\xi\equiv 1$ and hence we recovered an Archimedean copula.

\subsection{Proofs from Section~4}
\label{S:proofs}

\begin{proof}[Proof of Lemma~\ref{lem:positive}]
That $T(p)$ is non-negative is clear.  Moreover, we have 
\[ T(p) = 0 \iff \int \psi_\theta(P(y)) \ptrue(y) \,dy = 1 \quad \text{for Lebesgue-almost all $\theta$}. \]
It is easy to check, using the formula for $\psi_\theta$ and a change of variable, that $T(\ptrue) = 0$.  That $p=\ptrue$ is the unique (almost everywhere) solution follows from completeness of the normal mean family.  Indeed, if we make a change of variable $z=\Phi^{-1}(P(y))$, then 
\[ \int \psi_\theta(P(y)) \ptrue(y) \,\d y = \int \frac{\ptrue(P^{-1}(\Phi(z)))}{p(P^{-1}(\Phi(z)))} \nm(z \mid \theta, 1-\rho) \,\d z. \]
The condition that the integral above equals 1 for Lebesgue-almost all $\theta$ implies, by completeness of the normal mean family $\{\nm(\theta, 1-\rho): \theta \in \RR\}$, that the ratio in the integrand equals 1 for Lebesgue-almost all $z$, that is, 
\[ p(P^{-1}(\Phi(z))) = \ptrue(P^{-1}(\Phi(z))), \quad \text{for Lebesgue-almost all $z$}. \]
The claim that $p=\ptrue$ Lebesgue-almost everywhere follows immediately. 
\end{proof}

\begin{proof}[Proof of Lemma~\ref{lem:bound}]
Recall that the remainder term $R_n$ is given by 
\[ R_n = 2\alpha_n^2 \int \bigl[ c_\rho(P_{n-1}(y), P_{n-1}(Y_n)) - 1 \bigr]^2 \ptrue(y) \,\d y. \]
To get a handle on the conditional expectation $\E(R_n \mid \A_{n-1})$, it suffices to bound 
\[ \xi_n = \int \int c_\rho(P_{n-1}(y), P_{n-1}(y'))^2 \ptrue(y) \ptrue(y') \,\d y \,\d y'. \]
Using the formula \eqref{eq:copula.mix} for $c_\rho$ and Cauchy--Schwartz, we have 
\begin{align*}
c_\rho(u,v)^2 & = \Bigl\{ \int \psi_\theta(u) \psi_\theta(v) \nm(\theta \mid 0, \rho) \,\d\theta \Bigr\}^2 \\
& \leq \int \psi_\theta(u)^2 \nm(\theta \mid 0, \rho) \,\d\theta \int \psi_\theta(v)^2 \nm(\theta \mid 0, \rho) \,\d\theta. 
\end{align*}
Write $z_u = \Phi^{-1}(u)$ and note that, since 
\[ \int \nm(z_u \mid \theta, 1-\rho)^2 \nm(\theta \mid 0, \rho) \,\d\theta \propto \nm(z_u \mid 0, \tfrac{1+\rho}{2}), \]
and $\nm(z_u \mid 0, 1)^2 \propto \nm(z_u \mid 0, \frac12)$, we have 
\[ \int \psi_\theta(u)^2 \nm(\theta \mid 0, \rho) \, \d\theta \propto \frac{\nm(z_u \mid 0, \frac{1+\rho}{2})}{\nm(z_u \mid 0, \frac12)} \propto e^{\lambda z_u^2}, \]
where $\lambda = \rho/(1+\rho)$.  It follows from Inglot (2010, Theorem 2.1) that 
\[ |z_u| \leq \{-2 \log(u \wedge \bar u)\}^{1/2}, \quad u \in (0,1), \quad \bar u = 1-u. \] 
Applying this to the previous expression, we have that 
\[ c_\rho(u,v)^2 \lesssim ( u \wedge \bar u )^{-2\lambda} \cdot ( v \wedge \bar v )^{-2\lambda} \]
and, consequently, 
\[ \xi_n \lesssim \Bigl[ \int \{ P_{n-1}(y) \wedge \bar P_{n-1}(y) \}^{-2\lambda} \ptrue(y) \,\d y \Bigr]^2. \]
Consider the following trivial bounds:
\[ P_n(y) \geq P_0(y) \prod_{i=1}^n (1-\alpha_i) \quad \text{and} \quad \bar P_n(y) \geq \bar P_0(y) \prod_{i=1}^n (1-\alpha_i). \]
Substituting these bounds into the integral in the above display, we get the upper bound
\[ \Bigl\{ \prod_{i=1}^n (1-\alpha_i) \Bigr\}^{-2\lambda} \int \{P_0(y) \wedge \bar P_0(y)\}^{-2\lambda} p^\star(y) \,dy. \]
The integral is finite by \eqref{eq:integrable}.  The product term is upper-bounded by 
\[ \exp\Bigl\{\frac{2\lambda}{1-\alpha_1} \sum_{i=1}^n \alpha_i \Bigr\}, \]
and, since $\sum_{i=1}^n \alpha_i \sim a \log n$, in order for 
\[ \sum_n \alpha_n^2 \exp\Bigl\{\frac{4\lambda}{1-\alpha_1} \sum_{i=1}^n \alpha_i \Bigr\} < \infty \]
we need the constant $a > 0$ on the weights to satisfy 
\[ \frac{6\lambda a}{2-a} < 1. \]
The condition for $a$ in the statement of the lemma is based on solving this inequality, so this holds by assumption.  This implies that $\sum_n \E(R_n \mid \A_{n-1})$ converges $\Ptrue$-almost surely, proving the claim.  
\end{proof}

\section*{References}

\begin{description}

\item Bernardo, J.~M.~and Smith, A.~F.~M.~(1994). \emph{Bayesian Theory}. Wiley.




\item Carvalho, C.M., Lopes, H.F., Polson, N.G., and Taddy, M.A.~(2010). Particle learning for general mixtures. \emph{Bayesian Analysis}. {\bf 5}, 709--740.

\item Caudle, K.~A.~and Wegman, E.~(2009). Nonparametric density estimation of streaming data using orthogonal series.  \emph{Computational Statistics \& Data Analysis} {\bf 53}, 3980--3986.

\item Clayton, D.~G.~(1978). A model for association in bivariate life tables and its application in epidemiological studies of familial tendency in chronic disease incidence. \emph{Biometrika} \textbf{65}, 141--151. 

\item Creal, D.~(2012). A survey of sequential Monte Carlo methods for economics and finance. \emph{Econometric Reviews} {\bf 31}, 245--296.

\item de Finetti, B.~(1937).    La pr\'evision: ses lois logiques, ses
sources subjectives. \emph{Annals of the Institute of Henri Poincar$\acute{e}$}, {\bf 7}, 1--68.

\item Drovandi, C.~C., McGree, J.~M., and Pettitt, A.~N.~(2013).  Sequential Monte Carlo for Bayesian sequentially designed experiments for discrete data. \emph{Computational Statistics \& Data Analysis} {\bf 57}, 320--335.

\item Elgammal, A., Duraiswami, R. and Davis, L.S. (2003). Efficient kernel density estimation using the fast gauss transform with applications to color modeling and tracking. \emph{IEEE Transactions on Pattern Analysis and Machine Intelligence} {\bf 25}, 1499--1504.  

\item Escobar, M.~D.~(1988). \emph{Estimating the means of several normal populations by nonparametric estimation of the distribution of the means}. Unpublished PhD dissertation, Department of Statistics, Yale University.

\item Escobar, M.~D.~and West, M.~(1995). Bayesian density estimation and inference using mixtures. \emph{Journal of the American Statistical Association} \textbf{90}, 577---588.

\item Ferguson, T.~S.~(1973).  A Bayesian analysis of some nonparametric problems. \emph{Annals of Statistics} {\bf 1}, 209--230.

\item Gerber, M.~S.~(2014). Predicting crime using Twitter and kernel density estimation.  \emph{Decision Support Systems} {\bf 61}, 115--125.

\item Han, B., Comaniciu, D., Zhu, Y. and Davis, L.~S.~(2008). Sequential kernel density approximation and its application  to real time visual tracking. \emph{IEEE Transactions on Pattern Analysis and Machine Intelligence} {\bf 7}, 1186--1197.

\item Hewitt, E.~and Savage, L.~J.~(1955). Symmetric measures on Cartesian products. \emph{Transactions of the American Mathematical  Society} {\bf 80}, 470--501.

\item Hjort, N.~L., Holmes, C.~C., M\"uller, P., and Walker, S.~G.~(2010). \emph{Bayesian Nonparametrics}. Cambridge University Press.

\item Inglot, T.~(2010). Inequalities for quantiles of the chi-square distribution.  \emph{Probability and Mathematical Statistics} {\bf 30}, 339--351. 

\item Jara, A., Hanson, T., Quintana, F., M\"uller, P., and Rosner, G.~(2011). DPpackage: Bayesian Semi- and Nonparametric Modeling in R. \emph{Journal of Statistical Software}, {\bf 40}, 1--30. 

\item Kushner, H.~J.~and Yin, G.~G.~(2003).  \emph{Stochastic Approximation and Recursive Algorithms and Applications}, 2nd Ed.  Springer, New York.  

\item Lambert, C., Harrington, S., Harvey, C., and Glodjo, A.~(1999).  Efficient on-line nonparametric kernel density estimation.  \emph{Algorithmica} {\bf 25}, 37--57. 


\item Lo, A.~Y.~(1984). On a class of Bayesian nonparametric estimates
I. Density estimates. \emph{Annals of Statistics} {\bf 12}, 351--357.

\item Martin, R.~and Ghosh, J.~K.~(2008), Stochastic approximation and Newton's estimate of a mixing distribution. \emph{Statistical Science}, {\bf 23}, 365--382.

\item Martin, R.~and Tokdar, S.~T. (2009). Asymptotic properties of predictive recursion: robustness and rate of convergence.  \emph{Electronic Journal of Statistics} {\bf 3}, 1455--1472.  

\item Martin, R.~and Tokdar, S.~T. (2011). Semiparametric inference in mixture models with predictive recursion marginal likelihood.  \emph{Biometrika} {\bf 98}, 567--582.

\item Nakamura, Y.~and Hasegawa, O.~(2013). Robust fast online multivariate nonparametric density estimator. \emph{Lecture Notes in Computer Science} {\bf 8227}, 180--187.

\item  Nelsen, R.~B.~(1999). \emph{An Introduction to Copulas}. New York: Springer.

\item Newton, M.~A.~and Zhang, Y.~(1999). A recursive algorithm for nonparametric analysis with   missing data. \emph{Biometrika} {\bf 86}, 15--26.

\item Newton, M.~A.~(2002). On a nonparametric recursive estimator of the mixing distribution. \emph{Sankhya} {\bf 64}, 306--322.  

\item Raykar, V.~C., Duraiswami, R., and Zhao, L.~H.~(2010). Fast computation of kernel estimators. \emph{Journal of Computational and Graphical Statistics} {\bf 19}, 205--220.

\item Robbins, H.~and Siegmund, D.~(1971).  A convergence theorem for non negative almost supermartingales and some applications.  In \emph{Optimizing Methods in Statistics $(${P}roc. {S}ympos., {O}hio {S}tate {U}niv., {C}olumbus$)$}, 233--257.  Academic Press, New York. 

\item Sethuraman, J.~(1994). A constructive definition of Dirichlet priors. \emph{Statistica Sinica} {\bf 4}, 639--650.

\item Sheather, S.J. and Jones, M.C.~(1991). A reliable data-based bandwidth selection method for kernel density estimation. \emph{Journal of the Royal Statistical Society, Series B} {\bf 53}, 683--690.


\item Sklar, M.~(1959), Fonctions de r{\'e}partition {\'a} $n$ dimensions et leurs marges. Universit{\'e} Paris 8.

\item Taddy, M.A.~(2010). Autoregressive mixture models for dynamic spatial Poisson processes: Application to tracking intensity of violent crime. \emph{Journal of the American Statistical Association} {\bf 105}, 1403--1417.

\item Tokdar, S.~T., Martin, R., and Ghosh, J.~K.~(2009). Consistency of a recursive estimate of mixing distributions. \emph{Annals of Statistics} {\bf 37}, 2502--2522. 


\end{description}

\end{document}